\documentclass[12pt,a4paper]{article}

\usepackage[margin=1in]{geometry}
\usepackage{amsmath,amssymb,amsthm}
\usepackage{graphicx}
\usepackage[numbers,sort&compress]{natbib} 
\usepackage{hyperref}
\usepackage{booktabs} 
\usepackage{caption}  

\hypersetup{
    colorlinks=true,
    linkcolor=blue,
    filecolor=magenta,      
    urlcolor=cyan,
    citecolor=blue,
}

\title{\textbf{A Game-Theoretic Foundation for Bitcoin's Price: \\ A Security-Utility Equilibrium}}
\author{Liang Chen\thanks{Affiliation: Powsell Hong Kong Limited. This manuscript formalizes and substantially extends a conceptual framework developed by the author.}}
\date{August 2025}

\theoremstyle{plain}
\newtheorem{theorem}{Theorem}

\theoremstyle{definition}
\newtheorem{definition}{Definition}

\theoremstyle{remark}

\begin{document}

\maketitle

\begin{abstract}
\noindent The valuation of decentralized digital assets like Bitcoin presents a significant challenge to standard asset pricing paradigms, with explanations often defaulting to speculative belief dynamics. This paper develops a structural game-theoretic model where Bitcoin's price is determined in a general equilibrium of strategic interactions among security suppliers (miners), asset demanders (users and speculators), and potential attackers. We define and analyze a \textbf{Rational-Expectations Security-Utility Nash Equilibrium (RESUNE)}, a fixed point where the market-clearing price simultaneously induces a free-entry equilibrium in the market for computational power (hash rate), which in turn endogenously determines the ledger's security. Network security, defined as one minus the probability of a successful 51\% attack, is micro-founded using a global games model of coordination among potential attackers. This approach resolves the equilibrium selection problem and yields a unique, continuous security function dependent on price and hashrate. We prove the existence of a RESUNE and establish a sufficient condition for its uniqueness and local stability, which has a clear economic interpretation: the stabilizing direct effect of price on demand must dominate the potentially destabilizing indirect feedback from price to security. The framework generates sharp, falsifiable predictions regarding the systemic impact of exogenous shocks, such as the protocol-mandated halving of block subsidies. We specify a structural Vector Autoregression (VAR) model to test the theory's core mechanism, hypothesizing that a halving, \textit{ceteris paribus}, initiates a contraction in both hashrate and price. This model provides a coherent foundation for decomposing Bitcoin's value into components attributable to transactional utility, security guarantees, and speculative premia, and rationalizes the empirically observed unidirectional causality from price to hashrate.
\vspace{1cm}
\end{abstract}

\newpage

\section{Introduction}

The valuation of Bitcoin poses a fundamental puzzle to modern finance. Lacking a traditional cash-flow stream, a claim on assets, or a sovereign guarantee, its price determination mechanism remains contentious in the academic literature. While some researchers attribute Bitcoin's value primarily to narrative-driven speculation or behavioral factors, such explanations overlook the sophisticated industrial organization that secures the network. Bitcoin's value proposition is intrinsically linked to its security — the immutability of its ledger — which is not an exogenous parameter but an endogenous outcome of a strategic ecosystem \citep{budish2018, biais2019}. Miners expend real resources to provide computational security (hash rate) in exchange for rewards whose value is denominated in the very asset they secure. Asset holders, in turn, demand the asset based on its transactional utility and their perception of this endogenously produced security.

This paper formalizes this fundamental feedback loop. We posit that the price of Bitcoin is determined as a \textbf{Rational-Expectations Security-Utility Nash Equilibrium (RESUNE)}. This equilibrium is a fixed point where the asset price must satisfy two simultaneous conditions: (i) it must clear the market for the existing stock of coins, where demand is a function of price and perceived security; and (ii) it must be consistent with the "zero-profit" entry condition for the marginal miner, whose collective actions generate the very security that underpins the asset's demand. This approach provides a structural, micro-founded valuation framework that anchors price in the economics of security production, moving beyond reduced-form valuation heuristics such as the stock-to-flow model or network-to-transaction ratios.

Our primary contributions are fourfold. First, we define the RESUNE and formally model the strategic interdependence between profit-maximizing miners, utility-maximizing asset holders, and opportunistic attackers. This integrates disparate elements of the cryptocurrency ecosystem into a single, coherent general equilibrium framework.

Second, we deliver a micro-founded model of network security. The supply of hash rate arises from a competitive, free-entry market populated by a continuum of heterogeneous miners, reflecting the diverse cost structures observed in the mining industry \citep{cong2021}. More critically, the probability of a systemic 51\% attack is derived from a global game of coordination among potential attackers, following \citet{morris1998}. This avoids ad-hoc assumptions about attack costs and probabilities, embedding coordination risk directly into the security calculus and yielding a unique, continuous security function.

Third, we provide a rigorous theoretical analysis of the equilibrium. We prove the existence of a RESUNE using a fixed-point argument. We then establish a sufficient condition for its uniqueness and local stability. This condition has a clear economic interpretation, weighing the stabilizing force of standard price-elastic demand against a potentially destabilizing feedback loop where higher prices induce higher security, which in turn boosts demand. When this feedback channel is too strong, multiple equilibria and endogenous fragility can emerge, providing a structural explanation for the "death spirals" and systemic risk amplification observed in crypto markets.

Finally, our framework generates sharp, falsifiable empirical predictions. We utilize the Bitcoin protocol's scheduled "halving" of block subsidies as a clean, exogenous negative shock to the security supply side. Our model predicts that a halving, in isolation, is a fundamentally contractionary event for the network's economic value, leading to a new equilibrium with both a lower price and a lower hashrate. This prediction stands in contrast to popular market narratives of a "halving rally," which we argue are driven by concurrent, and potentially confounding, speculative demand shocks. We specify a structural Vector Autoregression (VAR) model to empirically disentangle these effects and test the theory's core mechanism. This approach also provides a theoretical rationalization for the robust empirical finding of a unidirectional causality from price to hashrate.

This paper proceeds as follows. Section 2 reviews the related literature. Section 3 presents the formal model. Section 4 analyzes the equilibrium's existence and uniqueness. Section 5 derives the comparative statics of a halving shock and outlines a testable empirical strategy. Section 6 discusses the model's implications for endogenous fragility and the long-term security budget. Section 7 concludes.

\section{Related Literature}

This paper builds on and contributes to three distinct strands of academic literature.

\subsection{The Economics of Proof-of-Work and Blockchain Security}
The foundation of Bitcoin's security is its Proof-of-Work (PoW) consensus mechanism \citep{nakamoto2008}. A critical literature has analyzed its economic limitations. \citet{budish2018} argues that because an attacker can potentially steal funds via a double-spend in addition to earning block rewards on their dishonest chain, there is an "attacker advantage" that makes securing the ledger against a 51\% attack prohibitively expensive. 

A second major concern is the long-term security budget. As the programmed block subsidy declines, the network must increasingly rely on transaction fees to incentivize miners. Several papers question the viability of a purely fee-based security model. \citet{easley2019} and \citet{huberman2021} model the transaction fee market as a strategic game, highlighting potential market failures. The "tragedy of the common chain" argument suggests that individual users have an incentive to free-ride on the security provided by others' fees, leading to an under-provision of security in equilibrium. Our model contributes to this literature by embedding the security production function directly into an asset pricing framework, allowing us to analyze how the asset's price and the security budget are jointly determined and how their relationship evolves as the block subsidy changes.

\subsection{Game-Theoretic and Microstructure Models of Cryptocurrency}
A growing literature uses game theory to model the strategic interactions within the crypto ecosystem. This includes dynamic games of miner competition \citep{saleh2021}, strategic behavior in transaction fee auctions \citep{easley2019}, and the industrial organization of mining pools, which act as intermediaries for individual miners \citep{cong2021}. \citet{biais2019} model the blockchain protocol as a stochastic game and show that while mining the longest chain is a Markov perfect equilibrium, multiple equilibria exist, including those with persistent forks. Our work shares this game-theoretic approach but integrates these elements—miner entry, user demand, and attacker strategy—into a unified general equilibrium model that solves for the asset's price, which is typically taken as exogenous in more focused microstructure models. Our use of a global games framework to model attack coordination is, to our knowledge, a novel application in this context, providing a rigorous microfoundation for the network's security level.

\subsection{Empirical Determinants of Cryptocurrency Value}
A large body of empirical work has sought to identify the drivers of cryptocurrency prices. Early work focused on factors analogous to traditional assets, often with limited success \citep{borri2019}. More recent studies have leveraged the transparency of blockchain data to construct "on-chain" factors. Proxies for network adoption and utility, such as the number of active addresses or transaction volumes, have been shown to have explanatory power.

Most relevant to our paper is the literature examining the relationship between price and hashrate, the latter being a direct measure of the computational resources dedicated to security. A robust finding across multiple studies is a unidirectional Granger causality running from price to hashrate, with lags typically ranging from one to six weeks. Our model provides a direct theoretical rationalization for this empirical regularity. The equilibrium price $P$ directly determines miner revenue, thus affecting the hash rate supply function $H(P)$. The reverse causality, from hashrate to price, operates through an indirect and potentially weaker channel: hashrate affects security, which in turn affects demand and thus price. Our framework formalizes this price-led dynamic.

\section{The Model}

The economy consists of four classes of risk-neutral agents: a continuum of potential miners, a continuum of asset users, a group of speculators, and a representative potential attacker. The total fixed supply of the asset (Bitcoin) is denoted by $Q$.

\begin{table}[!ht]
\centering
\caption{Model Variables and Parameters}
\label{tab:variables}
\begin{tabular}{@{}ll@{}}
\toprule
\textbf{Symbol} & \textbf{Definition} \\ \midrule
\multicolumn{2}{l}{\textit{Endogenous Variables}} \\
$P$ & Price of one unit of the crypto-asset (e.g., Bitcoin) \\
$H$ & Aggregate hash rate of the network \\
$\pi^A$ & Probability of a successful 51\% attack \\
$\varsigma$ & Perceived network safety, $\varsigma = 1 - \pi^A$ \\
$D_U$ & Asset demand from users \\
$D_S$ & Asset demand from speculators \\
$D$ & Total asset demand, $D = D_U + D_S$ \\
\midrule
\multicolumn{2}{l}{\textit{Exogenous Parameters and Functions}} \\
$Q$ & Fixed total supply of the asset \\
$B$ & Block subsidy reward (in units of the asset) \\
$\Phi$ & Average transaction fees per block (in units of the asset) \\
$c_i$ & Marginal cost of hash rate production for miner $i$ \\
$F(c)$ & CDF of miner costs \\
$gP$ & Gross payoff to the attacker from a successful attack \\
$kH$ & Cost to the attacker of mounting an attack \\
$\sigma_\varepsilon$ & Noise in the attacker's private signal about profitability \\
$\theta_U$ & User demand scale parameter (network effect) \\
$\epsilon$ & Price elasticity of user demand \\
$\theta_S$ & Speculator demand scale parameter \\
$\delta$ & Baseline risk cost for speculators \\ \bottomrule
\end{tabular}
\end{table}

\subsection{Miners and the Supply of Hash Rate}
A continuum of potential miners, indexed by $i$, is characterized by a constant marginal cost $c_i$ to produce one unit of hash rate. Costs are distributed according to a strictly increasing and continuous cumulative distribution function (CDF) $F(c)$ over the support $[c_{\min}, c_{\max}]$.

Miners are price-takers in a competitive market for hash rate. The revenue per unit of hash rate, $R$, is determined by the Bitcoin price $P$, the block subsidy $B$, and average transaction fees per block $\Phi$. This total block reward is shared among all miners in proportion to their contribution to the total hash rate $H$. Thus, the revenue per unit of hash is:
\begin{equation}
 R(P, H) = P \cdot \frac{B + \Phi}{H} \label{eq:revenue_model}
\end{equation}
A miner $i$ will be active if and only if $c_i \le R(P, H)$. Under free entry, miners enter until the marginal miner is indifferent. The cost of this marginal miner, $c_H$, must equal the revenue: $c_H = R(P, H)$. The aggregate hash rate $H$ is the measure of all active miners, $H = F(c_H)$. Substituting these conditions, the equilibrium supply of hash rate is implicitly defined by:
\begin{equation}
 F^{-1}(H) = P \cdot \frac{B + \Phi}{H} \label{eq:hashsupply_model}
\end{equation}
where $F^{-1}(\cdot)$ is the inverse CDF. Equation \eqref{eq:hashsupply_model} implicitly defines an aggregate hash rate supply function, $H = H(P)$, which is continuous and strictly increasing in price ($dH/dP > 0$).

\subsection{The Attacker and Endogenous Security}
Network security is its resilience to a 51\% attack. An attacker's gross payoff from an attack is $G(P) = gP$. The cost of attack is $K(H) = kH$. Following \citet{morris1998}, we assume the attacker is uncertain about the net profitability, $\Pi(P, H) = gP - kH$. They observe a private signal $s = \Pi(P, H) + \varepsilon$, where $\varepsilon \sim \mathcal{N}(0, \sigma_{\varepsilon}^2)$. This leads to a unique threshold strategy: attack if $s \ge s^*$. The attack probability is:
\begin{equation}
\pi^A(P, H) = \Pr(s \ge s^*) = 1 - \Psi\left(\frac{s^* - (gP - kH)}{\sigma_\varepsilon}\right) \label{eq:attackprob_model}
\end{equation}
where $\Psi(\cdot)$ is the standard normal CDF. We define **perceived network safety** as $\varsigma(P, H) = 1 - \pi^A(P, H)$. Safety $\varsigma$ is strictly increasing in $H$ and strictly decreasing in $P$.

\subsection{Users, Speculators, and Asset Demand}
Total demand $D(P, \varsigma)$ for the asset supply $Q$ comes from users ($D_U$) and speculators ($D_S$).
User demand is increasing in safety $\varsigma$ and decreasing in price $P$:
\begin{equation}
 D_U(P, \varsigma) = \theta_U \cdot \varsigma \cdot P^{-\epsilon} \label{eq:userdemand_model}
\end{equation}
where $\theta_U > 0$ captures network adoption and $\epsilon > 0$ is price elasticity.
Speculator demand is modeled in reduced form as increasing in safety and decreasing in price:
\begin{equation}
 D_S(P, \varsigma) = (\theta_S \cdot \varsigma - \delta) \cdot P^{-1} \label{eq:speculatordemand_model}
\end{equation}
where $\theta_S > 0$ reflects speculative interest and $\delta$ is a baseline risk parameter.
Total asset demand is $D(P, \varsigma) = D_U(P, \varsigma) + D_S(P, \varsigma)$.

\subsection{Rational-Expectations Security-Utility Nash Equilibrium (RESUNE)}
\begin{definition}
A \textbf{RESUNE} is a price $P^* > 0$ and hash rate $H^* > 0$ such that:
\begin{enumerate}
    \item \textbf{Miner Profit Maximization:} $H^* = H(P^*)$ satisfies Equation \eqref{eq:hashsupply_model}.
    \item \textbf{Asset Market Clearing:} $D(P^*, \varsigma(P^*, H^*)) = Q$.
\end{enumerate}
The equilibrium price $P^*$ is a solution to the fixed-point equation:
\begin{equation}
 D(P^*, \varsigma(P^*, H(P^*))) = Q \label{eq:equilibrium_model}
\end{equation}
\end{definition}

\section{Equilibrium Analysis}

\subsection{Existence of Equilibrium}
\begin{theorem}[Existence]
A RESUNE $(P^*, H^*)$ exists, provided that $\lim_{P\to 0} D(P, \varsigma(P, H(P))) > Q$ and $\lim_{P\to\infty} D(P, \varsigma(P, H(P))) = 0$.
\end{theorem}
\begin{proof}
Let the excess demand function be $\mathcal{Z}(P) = D(P, \varsigma(P, H(P))) - Q$. The function is continuous for $P>0$. The boundary conditions ensure $\mathcal{Z}(P)$ is positive for small $P$ and negative for large $P$. By the Intermediate Value Theorem, there must exist at least one $P^* \in (0, \infty)$ where $\mathcal{Z}(P^*) = 0$. The detailed proof is in Appendix A.
\end{proof}

\subsection{Uniqueness and Stability}
The uniqueness of the equilibrium depends on the slope of the aggregate demand curve. Differentiating demand with respect to price yields:
$$ \frac{dD}{dP} = \underbrace{\frac{\partial D}{\partial P} + \frac{\partial D}{\partial \varsigma}\frac{\partial \varsigma}{\partial P}}_{\text{Direct Price Effect (-) }} + \underbrace{\frac{\partial D}{\partial \varsigma}\frac{\partial \varsigma}{\partial H}\frac{dH}{dP}}_{\text{Indirect Security-Feedback Effect (+) }} $$
The Direct Price Effect is stabilizing (negative), while the Indirect Security-Feedback Effect is potentially destabilizing (positive).

\begin{theorem}[Uniqueness and Stability]
The RESUNE is unique and locally stable if the stabilizing Direct Price Effect dominates the destabilizing Indirect Security-Feedback Effect for all $P > 0$:
\begin{equation}
 \left| \frac{\partial D}{\partial P} + \frac{\partial D}{\partial \varsigma}\frac{\partial \varsigma}{\partial P} \right| > \frac{\partial D}{\partial \varsigma}\frac{\partial \varsigma}{\partial H}\frac{dH}{dP} \label{eq:uniqueness_model}
\end{equation}
\end{theorem}
\begin{proof}
This condition ensures that the total derivative $dD/dP$ is negative, so the aggregate demand curve is globally downward sloping and can intersect the vertical supply curve $Q$ only once. The detailed proof is in Appendix A.
\end{proof}
If this condition is violated, multiple equilibria can exist, creating endogenous fragility and the possibility of self-fulfilling "death spirals."

\section{Comparative Statics and Empirical Implications}

We analyze the "halving" event, where the block subsidy $B$ is cut by 50\%.

\subsection{Analysis of a Halving Shock}
\textit{Ceteris paribus}, a halving is a negative shock to miner revenue.
\begin{enumerate}
    \item \textbf{Direct Effect on Mining:} The hash supply function $H(P)$ shifts left. A higher price is now needed to sustain any given level of hashrate.
    \item \textbf{Effect on Security and Demand:} For any given price, the hashrate is lower, which reduces network safety $\varsigma$. This shifts the aggregate demand curve down and to the left.
    \item \textbf{New Equilibrium:} The new equilibrium $(P^{**}, H^{**})$ features a lower price and lower hashrate.
\end{enumerate}
This core prediction—that a halving is fundamentally contractionary for security and value—contrasts with popular market narratives. We argue that observed "halving rallies" are likely driven by confounding speculative demand shocks fueled by the narrative of scarcity, which can temporarily overwhelm the negative fundamental effect we model.

\subsection{A Falsifiable VAR Specification}
We propose a structural VAR model with the vector of endogenous variables $\mathbf{y}_t = [\log(P_t), \log(H_t), \log(\Phi_t)]'$, measured at a weekly frequency. The halving is an identified structural shock to miner revenue.

\begin{table}[!ht]
\centering
\caption{Predicted Impulse Responses to a Halving Shock}
\label{tab:var_predictions}
\begin{tabular}{@{}lccc@{}}
\toprule
\textbf{Variable} & \textbf{Short-Run} & \textbf{Long-Run} & \textbf{Rationale} \\ \midrule
$\log(H_t)$ & Negative & Negative & Immediate drop in revenue forces marginal miners offline. \\
\addlinespace
$\log(P_t)$ & Negative & Negative & Lower security reduces asset demand and price. \\
\addlinespace
$\log(\Phi_t)$ & Ambiguous & Positive & Fees must eventually rise to compensate for lost subsidy. \\ \bottomrule
\end{tabular}
\end{table}

\noindent\textbf{Hypothesis:} An identified negative shock to miner revenue will generate the impulse responses in Table \ref{tab:var_predictions}. An empirical finding that price robustly rises post-halving without a corresponding surge in fees or an identifiable exogenous demand shock would challenge this model's core mechanism.

\section{Discussion}

\subsection{Endogenous Fragility and Systemic Risk}
The model provides a structural foundation for the extreme fragility of crypto-asset markets. When the stability condition is violated, the system is prone to self-fulfilling shifts between "good" (high-price, high-security) and "bad" (low-price, low-security) equilibria. This feedback loop can amplify negative shocks, providing a micro-foundation for market crashes and "deleveraging spirals" observed in cryptocurrency markets.

\subsection{The Long-Term Security Budget}
As $B \to 0$, miner revenue becomes entirely dependent on fees $\Phi$. For the network to remain secure, the fee market must become sufficiently large and robust to incentivize a hashrate $H^*$ that provides adequate security to justify the price $P^*$. Whether the demand for block space can generate such fees is a critical open question for Bitcoin's future viability.

\section{Conclusion}

This paper has developed a game-theoretic model where Bitcoin's price is determined as a RESUNE, structurally anchored by the endogenously supplied level of network security. The framework formalizes the price-hashrate nexus, provides a clear, testable mechanism for systemic shocks like the halving, and offers a lens for analyzing the network's long-term economic viability and inherent fragility. By providing a disciplined and falsifiable foundation for understanding decentralized value creation, this work moves the conversation beyond purely narrative-based explanations toward a more rigorous economic analysis.

\newpage

\bibliographystyle{plainnat}

\appendix
\section{Proofs}

\subsection{Proof of Theorem 1 (Existence)}
As outlined in the main text, the proof relies on the Intermediate Value Theorem applied to the excess demand function $\mathcal{Z}(P) = D(P, \varsigma(P, H(P))) - Q$.

The function $H(P)$ is defined implicitly by $H \cdot F^{-1}(H) = P(B+\Phi)$. The left-hand side is a strictly increasing function of $H$. The right-hand side is strictly increasing in $P$. Therefore, for each $P > 0$, there is a unique $H > 0$. By the Implicit Function Theorem, $H(P)$ is a continuous function.

The functions $\varsigma(P, H)$ and $D(P, \varsigma)$ are continuous by construction. Therefore, the composite function $\mathcal{Z}(P)$ is continuous for $P \in (0, \infty)$.

The boundary conditions are:
\begin{enumerate}
    \item As $P\to 0$, $P(B+\Phi) \to 0$, which implies $\lim_{P\to 0} H(P) = 0$. The price effect on user demand ($P^{-\epsilon}$) dominates, causing demand to become arbitrarily large. Thus, $\lim_{P\to 0} D(P, \varsigma(P, H(P))) = \infty > Q$.
    \item As $P \to \infty$, both $P^{-\epsilon} \to 0$ and $P^{-1} \to 0$. Therefore, $\lim_{P\to\infty} D(P, \varsigma(P, H(P))) = 0 < Q$.
\end{enumerate}
Since $\mathcal{Z}(P)$ is continuous and changes sign from positive to negative over the interval $(0, \infty)$, there must exist at least one $P^* \in (0, \infty)$ such that $\mathcal{Z}(P^*) = 0$.

\subsection{Proof of Theorem 2 (Uniqueness and Stability)}
As argued in the main text, the uniqueness of the equilibrium is guaranteed if the aggregate demand function $D_{agg}(P) \equiv D(P, \varsigma(P, H(P)))$ is monotonically decreasing in $P$. The slope is given by:
$$ \frac{dD_{agg}}{dP} = \frac{\partial D}{\partial P} + \frac{\partial D}{\partial \varsigma}\left( \frac{\partial \varsigma}{\partial P} + \frac{\partial \varsigma}{\partial H}\frac{dH}{dP} \right) $$
The signs of the partial derivatives are: $\frac{\partial D}{\partial P} < 0$, $\frac{\partial D}{\partial \varsigma} > 0$, $\frac{\partial \varsigma}{\partial P} < 0$, $\frac{\partial \varsigma}{\partial H} > 0$, and $\frac{dH}{dP} > 0$.
Substituting these signs, we can rewrite the total derivative as:
$$ \frac{dD_{agg}}{dP} = \underbrace{\left( \frac{\partial D}{\partial P} + \frac{\partial D}{\partial \varsigma}\frac{\partial \varsigma}{\partial P} \right)}_{(-)} + \underbrace{\left( \frac{\partial D}{\partial \varsigma}\frac{\partial \varsigma}{\partial H}\frac{dH}{dP} \right)}_{(+)} $$
The first term (Direct Price Effect) is negative. The second term (Indirect Security-Feedback Effect) is positive. For the total derivative to be negative, the absolute value of the first term must be greater than the second term:
$$ \left| \frac{\partial D}{\partial P} + \frac{\partial D}{\partial \varsigma}\frac{\partial \varsigma}{\partial P} \right| > \frac{\partial D}{\partial \varsigma}\frac{\partial \varsigma}{\partial H}\frac{dH}{dP} $$
This is the condition in Equation \eqref{eq:uniqueness_model}. If this holds globally, $dD_{agg}/dP < 0$, and the strictly decreasing function can cross the constant $Q$ at most once, proving uniqueness.

\end{document}